\documentclass{article}% Math and Physical Sciences Reference Style
\usepackage{amsmath,latexsym,amssymb,graphics,amsthm,amscd}
\usepackage[pdftex]{graphicx}
\usepackage{epic}
\usepackage{subfig}
\usepackage{fancyhdr}
\usepackage{fancybox}
\usepackage{boxedminipage}
\usepackage{shadow}
\usepackage{anysize}
\usepackage{extramarks}
\usepackage{fancybox}
\usepackage{multicol}
\usepackage{comment}
\usepackage{xcolor}
\usepackage{newtxmath}
\usepackage{tikz}
\usepackage{calc}
\usepackage{graphicx}
\usepackage{imakeidx}
\usepackage{diagbox}
\usepackage{caption}

\newtheorem{theorem}{Theorem}
\newtheorem{proposition}[theorem]{Proposition}
\newtheorem{coro}[theorem]{Corollary}
\newtheorem{lemma}[theorem]{Lemma}

\newcommand{\mb}{\mathbb}

\newcommand{\htau}{\hat{\tau}}

\begin{document}

\title{Causal completions as Lorentzian pre-length spaces}

\author{Luis Ake Hau, Saul Burgos and Didier A Solis}
\date{}

%\author[1]{\fnm{Luis} \sur{Ake Hau}}\email{luis.ake@correo.uady.mx}

%\author[1]{\fnm{Saul} \sur{Burgos}}\email{saul.burgos@alumnos.uady.mx}

%\author*[1]{\fnm{Didier A.} \sur{Solis}}\email{didier.solis@correo.uady.mx}

\maketitle

%\affil*[1]{\orgdiv{Facultad de Matem\'aticas}, \orgname{Universidad Aut\'onoma de Yucat\'an}, \orgaddress{\street{Anillo Perif\'erico 13615}, \city{M\'erida}, %\postcode{100190}, 
%\state{Yucat\'an}, \country{M\'exico}}}

\abstract{In this work we revisit the notion of the (future) causal completion of a globally hyperbolic spacetime and endow it with the structure  of a Lorentzian pre-length space. We further carry out this construction for a certain class of generalized Robertson-Walker spacetimes.}

\bigskip


\medskip


\section{Introduction}\label{sec1}

There is no doubt that Roger Penrose is one of the precursors of the mathematical foundations of causal  theory in General Relativity. His famous notes on causality \cite{penrose0} have become a classic textbook and a solid starting point for graduate students and researchers alike.  In this regard, one of Penrose's great achievements in Mathematical Relativity is the formal study of the asymptotic structure of spacetime, which arose with the introduction of the conformal compactification and Penrose diagrams \cite{zero}. It is through the notion of conformal infinity that the notion of a black hole is abstracted, thus allowing the development of  geometric and causal theoretic methods in their analysis. In spite of its success and widespread use, the conformal compactification approach to the asymptotic structure of spacetime has some downsides. Most notably, given a  spacetime, there is not a straightforward way to decide whether it admits a conformal compactification or not, or even in the affirmative case, a canonical way of constructing such a compactification. This issue was also tackled by Penrose. In their seminal work,  Geroch, Kronheimer and Penrose \cite{gkp} provided an alternative way to deal with the structure at infinity of a spacetime that relies exclusively in the causal structure of a distinguishing spacetime. Their construction follows in spirit the classical constructions in elementary geometry, where an ideal point (or point at infinity) is attached to a family of curves having a common end. Instead of using parallel rays (as is the case say, in hyperbolic geometry) they considered the causal structure of spacetime and declared that two curves have a common point at future (past) infinity if their corresponding chronological past (future) sets agree. Though elegant and simple at first sight,  the construction of the so called \emph{causal boundary} of a spacetime involves many subtleties, specially when trying to endow it with some additional structure, like a topology or causal relations (see \cite{flores:final,waveS} and references there in for an up to date account). Recently, significant results have been accomplished in this regard, most notably the introduction of a notion of black hole based on the causal boundary \cite{bhcausal}.

On the other hand, we have witnessed in the past few years a surge in the use of non-smooth geometric methods in Mathematical Relativity. Diverse settings as cone structures \cite{samann03,mingucone}, $C^0$ metrics \cite{galloling,grantc0, hevelin02,ling01,sbierski} and Lorentzian length spaces --to mention just a few-- have proven useful in exploring scenarios where (metric) smoothness can not be guaranteed, as the ones linked to recent observations \cite{ebh,ligo}. As a matter of fact, the use of non-smooth methods is not new. In the context of causality, Penrose and Kronheimer were the first to provide an abstract framework that does not require a metric structure at all \cite{kpcs}. Their notion of \emph{causal space} lays at the foundations of the theory of Lorentzian pre-length spaces first introduced by Kunzinger and S\"amman \cite{KSlls}. The purpose of this work is to present the future (or past) causal completion of a globally hyperbolic spacetime as a Lorentzian pre-length space, thus adding an interesting source of examples to this rapidly growing field \cite{samann02,hevelin,beran02,beran01,solismontes,hendike,kunzinger02}.

This work is organized as follows. In section \ref{sec:prelim}  we establish the basic facts about the causal completion and Lorentzian pre-lenth spaces, as well as the notation that will be used throughout this work. In section \ref{sec:gen} we prove that the future (past) causal completion admits a natural Lorentzian pre-length structure. Finally, in section \ref{sec:warped} we exhibit this structure in a class of warped product spacetimes.

%%%%%%%%%%%%%%%%%%
%%%%%%%%%%%%%%%%%%. PRELIMINARES
%%%%%%%%%%%%%%%%%%

\section{Preliminaries}\label{sec:prelim}

\subsection{Causal completions}

Let $(M,g)$ be a strongly causal spacetime and $\ll$, $\leq$ its usual chronological and causal relations, that is,  $p \ll q$ if and only if there exists a smooth future-directed timelike curve that joins $p$ with $q$, whereas $p \leq q$ if and only if there exists a smooth future-directed causal curve between these points. We define the \emph{chronological} (\emph{causal}) \emph{future} and \emph{past sets} in the standard way:  
\begin{eqnarray*}
I^{+}(p) =\{q \in M \mid p \ll q\}, &\qquad& J^{+}(p)=\{q \in M \mid p \leq q\},\\
I^{-}(p) =\{q \in M \mid q \ll p\}, &\qquad& J^{-}(p)=\{q \in M \mid q \leq p\}.
\end{eqnarray*}
A sequence of points $\{x_{n}\}$ is called a \emph{future-directed chain} if $x_{n} \ll x_{n+1}$ for all $n$ and \emph{past directed} if $x_{n+1} \ll x_{n}$ for all $n \in \mathbb{N}$. Moreover, it will be called \emph{inextensible} if $\{x_{n}\}$ is not convergent. A subset $P \subseteq M$ is called \emph{past set} if $P=I^{-}(P)$ and for a subset $S \subset M$ the \emph{common past} is defined as 
\[
\downarrow S=I^{-}(\{q \in M \mid q \ll s \ \forall s \in S\}).
\] 
We can define the notions of \emph{future sets} and the \emph{common future} $\uparrow S$ in a time dual way. We will say that $P$ is an \emph{indecomposable past set} (or IP for short)  if $P$ cannot be written as $P=P_{1} \cup P_{2}$ where $P_{1}, P_{2} \subset P$ are disjoint proper past subsets of $P$.  As it turns out, there are only two classes of indecomposable past sets: $(1)$ the chronological past of points $I^{-}(p)$, which will be called \emph{proper idecomposable past sets}  (PIP) and $(2)$ the chronological past of inextensible future-directed chains $I^{-}(\{x_{n}\})$ which will be called \emph{terminal idecomposable past sets} (TIP).

The \emph{future causal completion} $\hat{M}$ of $(M,g)$ is the set of all idecomposable past sets (IPs). Observe that being $(M,g)$ strongly causal, then it is past distinguishing. Hence if $I^-(p)=I^-(q)$, then $p=q$. Thus, we have that any point $p \in M$ determines a unique PIP. The \emph{future causal boundary} $\hat{\partial}M$ is then identified with the TIPs. Thus we have

\begin{equation*}
	\begin{split}
IPs \equiv & PIPs \cup TIPs, \\ 	
\hat{M}\equiv & M \cup \hat{\partial}M.
\end{split}
\end{equation*}

In a similar way we can define the \emph{past causal completion} as $\check{M} \equiv IFs$ and the \emph{past causal boundary} as  $\check{\partial}M=TIFs$. That is,

\begin{equation*}
	\begin{split}
		IFs \equiv & PIFs \cup TIFs \\ 	
		\check{M}\equiv & M \cup \check{\partial}M
	\end{split}
\end{equation*}   

Since both $\hat{M}$ and $\check{M}$ include a copy of the spacetime $(M,g)$, it is natural trying to define the total causal completion  of $(M,g)$ as $\hat{M}\cup\check{M}$, where the PIP $I^-(p)$ is identified with the PIF $I^+(p)$ ---since both represent the point $p\in M$. However, it was observed early on that such a construction may lead to inconsistencies, as in some cases further identifications on the boundaries $\hat{\partial}M$ and $\check{\partial}M$ ought to take place \cite{gkp,marolf,waveS,Sz}.  Thus finding an approach to constructing the causal boundary that works in full generality proved to be a delicate task, which was completed only recently \cite{flores:final}.   In order to avoid such intricacies, we will focus only on the future causal completion $\hat{M}$.

Notice that, as the notation suggest, the  relation $\hat{\ll}$ on $\hat{M}$ given by
\begin{equation}
	\label{eq2}
	\begin{split}
	P \hat{\ll} Q &\Leftrightarrow \exists q \in Q\setminus P {\text{ such that }} P \subset I^{-}(q)  \\
%	P \hat{\leq} Q & \Leftrightarrow P \subset Q.
	\end{split}
\end{equation}
is indeed transitive. In fact, if $P \hat{\ll} Q \hat{\ll} R$ there exist $q \in Q \setminus P$ and $r \in R \setminus Q$ such that $P \subset I^{-} (q) \subset Q \subset I^{-}(r) \subset R$, thus $P \hat{\ll} R$. 

Hence, we can think of $\hat{\ll}$ as providing a chronological structure on $\hat{M}$ \cite{harrischro}. The chronological future and past sets so induced will be denoted by $\hat{I}^+(P)$ and $\hat{I}^-(P)$.

Further, we can endow $\hat{M}$ with a sequential topology which is compatible with the chronology $\hat{\ll}$ just defined.\footnote{Recall that on any spacetime $(M,g)$ the chronological sets $I^+(p)$, $I^-(p)$ are open in the manifold topology.} Thus, consider the  limit operator $\hat{L}_{chr}$ over the sequences of past sets given by
\begin{equation}
\label{eq3}
P \in \hat{L}_{chr}(\{P_{n}\}\}) \Leftrightarrow P \subset LI(\{P_{n}\}) \text{ and it is maximal in } LS(\{P_{n}\}).\footnote{The theoretical inferior and superior limits of subsets are defined as $LI(\{P_{n}\})=\bigcup_{n=1}^{\infty}\bigcap_{m=n}^{\infty} P_{m}$ and $LS(P_{n})=\bigcap_{n=1}^{\infty} \bigcup_{m=n}^{\infty} P_{m} $, respectively.}
\end{equation}
and define the \emph{future chronological topology} $\hat{\mathcal{T}}_{chr}$ by its closed subsets as follows: a subset $C \subset \hat{M}$ is closed if and only if for any sequence $\sigma \subset C$ the contention $\hat{L}_{chr}(\sigma) \subset C$ holds.  Thus we have the following result (see the proof of Thrm. 3.27 in \cite{flores:final})
\begin{theorem} %\cite[Theorem 3.27]{FHSFinalDef}
\label{teo:propcomple}
Let $(M,g)$ be a strongly causal spacetime and $\hat{M}$ its future causal completion endowed with the chronological structure induced by \eqref{eq2} and the topology induced from the chronological limit \eqref{eq3}. Then:
	\begin{itemize}
		\item[(i)] The inclusion $M \hookrightarrow \hat{M}$ is continuous. Moreover, the restriction of the chronological topology to $M$ is the manifold topology.%		
		\item[(ii)]  The future causal completion is complete: any future directed chain $\{P_n\}$ in $\hat{M}$ converges in $\hat{\mathcal{T}}_{chr}$. In particular, any future inextensible timelike curve $\gamma$ on $M$ %(resp. any inextensible chronological sequence $\{x_n\}_{n}$ on $V$)
		has an endpoint in $\hat{M}$.
		
		\item[(iii)] The sets $\hat{I}^+(P)$ and  $\hat{I}^-(P)$ are open for all $P \in \hat{M}$.
		
		\item[(iv)] $(\hat{M},\hat{\mathcal{T}}_{chr})$ is a $T_1$ topological space.
		
	\end{itemize}
\end{theorem}

Since there are examples in which $(\hat{M},\hat{\mathcal{T}}_{chr})$ is not Hausdorff, we can not aim to furnish a manifold structure on it for such cases. However, some extra structure can be added in particular cases. For instance, a linear connection can be defined on $\hat{M}$ provided that $(M,g)$ is static and spherically symmetric  \cite{harriscon}.

\subsection{Lorentzian pre-length spaces}

In their remarkable paper \cite{kpcs}, Kronheimer and Penrose developed the notion of a causal space by abstracting the fundamental properties of the chronological and causal relations. Such an axiomatic approach to causality has proven useful in many circumstances, for instance in the causet approach to quantum gravity \cite{surya}. In recent times, the search for applying synthetic geometrical methods to mathematical relativity sparkled a renewed interest in developing causality in an abstract setting. Lorentizan pre-length spaces are a refinement of the notion of causal spaces that incorporates a time distance function, thus providing an analog for the length structure that serves as a building block for the well established synthetic theory in metric spaces \cite{burago}.  

As in the seminal work by Kunzinger and S\"amann \cite{KSlls} we define a Lorentzian pre-length space  $(X,d,\ll,\leq, \tau)$ as a metric space $(X,d)$ along with two relations $\ll$, $\leq$,  ---named \emph{chronological} and \emph{causal}--- and a \emph{time separation function} $\tau: X\times X\to [0,\infty]$ that satisfies the following axioms:
\begin{enumerate}
\item $\leq$ is a pre-order,
\item  $\ll$ is a transitive relation contained in $\leq$,
\item $\tau$ is a lower semi-continuous function ---with respect to $d$--- satisfying
\begin{itemize}
    \item $\tau(x,z)\geq \tau(x,y) + \tau(y,z)$ for all $x\leq y \leq z$,
    \item $\tau(x,y)>0$ if and only if $x\ll y$.
\end{itemize}
\end{enumerate}
  
 As an immediate consequence of the definition we have two of the most important features of the causal structure of a spacetime: (i) the chronological sets $I^+(p)$, $I^-(p)$ are open, and (ii) either $x\leq y\ll z$ or $x\ll y\leq z$ implies $x\ll z$.\footnote{This is commonly known as the \emph{push up property}.}  These properties, along some additional structure enable us to build a causality ---an even a causal hierarchy--- that closely resembles the usual causal structure of a smooth spacetime \cite{ACS,KSlls}.  Examples of Lorentzian pre-length spaces include a wide variety of structures, like cones \cite{samann03,mingucone}, causally plain $C^0$ spacetimes \cite{grantc0}, contact structures \cite{hendike} and Lorentzian taxicab-type spaces \cite{solismontes}.

%%%%%%%%%%%%%%%%%%%
%%%%%%%%%%%%%%%%%%%
%%%%%%%%%%%%%%%%%%%

\section{Lorentzian pre-length space structure on $\hat{M}$}\label{sec:gen}

As we discuss in the previous section, the topological space $(\hat{M},\hat{\mathcal{T}}_{chr})$ might not be metrizable in general. Thus, in order to construct a Lorentzian pre-length space structure on it, a refinement on the topology is in order.  In a recent work \cite{costaflores:newtopology}, Costa, Flores and Herrera studied the so called \emph{closed limit topology} $\hat{\mathcal{T}}_{c}$ (or CLT for short). The Hausdorff limit operator
\begin{equation}
\hat{L}_{H}(\{P_{n}\})=\{P \in \hat{M} \mid P=\liminf(\{P_n\})=\limsup(\{P_{n}\})\}
\end{equation}
generates the closed sets in $\hat{\mathcal{T}}_{c}$ as follows: a subset $C \subset \hat{M}$ is closed if and only if $\hat{L}_{H}(\{P_{n}\}) \subset C$ for every sequence $\{P_{n}\} \subset C$.\footnote{Observe that  $\# L_{H}(\sigma) \in \{0,1\}$ and thus $\mathcal{T}_{c}$ is a {\it first order topology} (see \cite{covering}). In other words, the following equivalence holds 
\[
\{P_{n}\} \rightarrow_{\hat{\mathcal{T}}_{c}} P \Leftrightarrow \hat{L}_{H}(\{P_{n}\})=\{P\}.
\] }
We now summarize the relevant features of the CLT topology in globally hyperbolic spacetimes (refer to Thrms 4.1 and 4.2 in \cite{costaflores:newtopology})
\begin{theorem}\label{teo:clt}
If $(M,g)$ is globally hyperbolic, then, the following statements hold for the topological space $(\hat{M},\hat{\mathcal{T}}_{c})$:
\begin{itemize}
	\item[$(i)$] The natural inclusion $i: M \rightarrow \hat{M}$ given by $i(p)=I^{-}(p)$ is an open continuous map. Moreover, $i(M)$ is an open dense subset of $\hat{M}$, the induced topology on $M$ is the manifold topology, $\hat{\partial} M$ is closed and $(\hat{M},\hat{\mathcal{T}}_{c})$ is second countable.
	
	\item[$(ii)$] The chronological sets $\hat{I}^{\pm}(P)$ are open subsets for all $P \in \hat{M}$.
	
	\item[$(iii)$] Any future directed chain in $\{P_{n}\} \subset \hat{M}$ converges in $\hat{\mathcal{T}}_{c}$
		\item[$(iv)$] The topological space $(\hat{M},\hat{\mathcal{T}}_{c})$ is metrizable.
\end{itemize}
\end{theorem}

Notice that in general, the chronological topology is coarser than $CLT$. Moreover, they coincide when their corresponding limit operators agree. The next result provides necessary and sufficient conditions for that (see \cite[Thrm. 5.3]{costaflores:newtopology})

\begin{theorem}\label{teo:opagree}
Let $\hat{M}$ be a future completion endowed with the limit operators $L_{H}$ and $\hat{L}_{chr}$.
Both limit operators coincide if, and only if, the chronological topology is Hausdorff.
\end{theorem}

Recall $\hat{M}$ is endowed with the chronological relation $\hat{\ll}$ given by \eqref{eq2}. Thus, $\hat{\ll}$ naturally induces an associated causal relation ${\leq}_\circ$ in $\hat{M}$  (see \cite[Definition 2.22]{sanchezminguzzi}) by setting 
\begin{equation*}
P {\leq}_\circ Q \qquad \Leftrightarrow \qquad \hat{I}^{-} (P) \subset \hat{I}^{-} (Q) \quad \text{and} \quad \hat{I}^{+} (Q) \subset \hat{I}^{+} (P).
\end{equation*}
By construction, this relation is reflexive, transitive and contains $\hat{\ll}$, thus it is a causal relation according to the definition of a pre-length space.

The causal relation $\leq_\circ$ though constructed by a standard procedure, it may seem artificial at first sight. The simple contention of past sets provides right away a pre-order in $\hat{M}$. That is,  the relation $\hat{\leq}$ on $\hat{M}$ given by
$$P \hat{\leq} Q \Leftrightarrow P \subset Q.$$
might be considered as well \cite{outline}.

The following well known result will be used extensively from now on. We include its proof here for completeness (see for example \cite[Proposition 2.11]{costaflores:newtopology})

\begin{proposition}\label{lemadesigualdad}
Let $(M,g)$ is a globally hyperbolic spacetime. If $P \in \hat{\partial} M$ then for any $p\in M$ we have $P\hat{\nleq} I^-(p)$.  \footnote{Equivalently, if $P,Q \in \hat{M}$ are IPs such that $P \in \hat{\partial} M$ and $P \hat{\leq} Q$, then $Q \in \hat{\partial} M$. }
\end{proposition}

\begin{proof}
Suppose by contradiction that $P \subset I^{-}(p)$ and consider $\gamma:[a,b) \rightarrow M$ an inextensible future directed timelike curve such that $P=I^{-}(\gamma)$. Since $P$ is a TIP 
then we must have that $P$ is a proper subset of $I^{-}(p)$, thus, if we take $\gamma(a)$ then we must have that $\gamma(s) \in J^{+}(\gamma(a)) \cap I^{-}(p) \subset J^{+}(\gamma(a)) \cap J^{-}(p)$ 
for all $s \in (a,b)$ and this gives an inextensible timelike curve imprisoned in a compact subset thus contradicting strong causality. 
\end{proof}

We now show that the causal relations $\leq_\circ$ and $\hat{\leq}$ coincide when $(M,g)$ is globally hyperbolic.

\begin{proposition}\label{relacionescausalescoinciden}
Let $(M,g)$ be a globally hyperbolic spacetime. For all $P,Q \in \hat{M}$, $P\hat{\leq} Q $ if and only if $P {\leq}_\circ Q$.
\end{proposition}

\begin{proof}
Suppose first that $P \hat{\leq} Q$, i.e. $P \subset Q$. If $Q$ is a TIP then by Proposition \ref{lemadesigualdad} $\hat{I}^{+} (Q) = \emptyset$ and we have $\hat{I}^{+} (Q) \subset \hat{I}^{+} (P)$. If $Q = I^{-} (q)$ is a PIP, consider $R  \in \hat{I}^{+} (Q)$, then $Q \hat{\ll} R$, which means there exists $r \in R \setminus Q$ such that $Q \subset I^{-} (r)$ and since $ P \subset Q \subset I^{-} (r)$ with $r \in R \setminus P$ we have $R \in \hat{I}^{+} (P)$. Thus $\hat{I}^{+} (Q) \subset \hat{I}^{+} (P)$.

Similarly, if $R \in \hat{I}^{-} (P)$ then $R \hat{\ll} P$, and there exists $r \in P \setminus R$ with $R \subset I^{-}(p)$. Since $P \subset Q$, $p$ also belongs to $Q$ and we get $R \hat{\ll} Q$, that is, $R \in \hat{I}^{-} (Q)$. This means that $\hat{I}^{-} (P) \subset \hat{I}^{-} (Q)$ and we have $P {\leq}_\circ Q$. 

On the other hand, let us suppose $P {\leq}_\circ Q$, that is,
$$\hat{I}^{+} (Q) \subset \hat{I}^{+} (P) \qquad \text{ and } \qquad \hat{I}^{-} (P) \subset \hat{I}^{-} (Q).$$
Take $p_0 \in P$, then there exists $q_0 \in P \setminus I^{-} (p_0)$ with $I^{-} (p_0) \subset I^{-} (q_0)$ therefore $I^{-} (p_0) \hat{\ll} P$, that is, $I^{-} (p_0) \in \hat{I}^{-}(P) \subset \hat{I}^{-}(Q)$. Thus, there exists $q \in Q\setminus I^{-}(p_0)$ con $I^{-} (p_0) \subset I^{-} (q) \subset Q$. If $\{p_n\}$ is a future directed chain generating  $I^{-}(p_0)$ we have that $p_n \in I^{-}(q)$ for $n$ large enough and then $p_0 \in \overline{I^{-}(q)} = J^{-}(q)$, since $M$ is causally simple. In consequence, $p_0 \leq q \ll q_m,$ for some $m$ large enough, where $\{q_m\}$ is a future-directed chain that generates $Q = I^{-}(\{q_m\})$. Then, $p_0 \in Q$ and $P \subset Q$, which by definition means $P \hat{\leq} Q$. 
\end{proof}

As an immediate consequence, we have the following

\begin{coro}\label{hatMescausal}
Let $(M,g)$ be a globally hyperbolic spacetime. Then the relation $\hat{\leq}$ is a partial order on $\hat{M}$.\footnote{This amounts to saying that the causal space $(\hat{M},\hat{\ll},\hat{\leq})$ satisfies the causality axiom. }
\end{coro}

As expected, the relations $\hat{\ll}$ and $\hat{\leq}$ naturally extend those of $M$.

\begin{proposition}\label{relacionesextendidas}
Let $(M,g)$ be globally hyperbolic. For all $p,q \in M$, $I^{-} (p) \hat{\ll} I^{-} (q)$ if and only if $p \ll q$. Moreover, $I^{-} (p) \hat{\leq} I^{-} (q)$ if and only if $p \leq q$.
\end{proposition}

\begin{proof}
Assume first that $I^{-} (p) \hat{\ll} I^{-}(q)$. This implies, by definition, that there is a $r \in I^{-}(q) \setminus I^{-} (p)$ such that $I^{-} (p) \subset I^{-} (r)$. Then $p \ll r \ll q$ which implies $p \ll q$.

Conversely, if $p \ll q$ then $p \neq q$ because $(M,g)$ is distinguishing. Consider $r \in M$ with $p \ll r \ll q$. Thus, $r \in I^{-}(q) \setminus I^{-} (p)$ with $I^{-} (p) \subset I^{-} (r)$, that is, $I^{-} (p) \hat{\ll} I^{-} (q)$. This concludes the first part of the proof.

Now let us suppose that $I^{-} (p) \hat{\leq} I^{-} (q)$, that is, $I^{-} (p) \subset I^{-} (q)$. For $p \in M$, we take $\{p_n\}$ a future-directed chain that generates $I^{-} (p)$, that is, a sequence with $p_n \ll p_{n+1}$ for all $n \in \mathbb{N}$ and $p_n \to p$. Notice that for all $n \in \mb{N}$ we have that $p_n \ll p$, that is, $p_n \in I^{-} (p)$ and so $p_n \in I^{-}(q)$. Therefore $p_n \to p \in \overline{I^{-}(q)} = J^{-}(q)$ because $(M,g)$ is causally simple. Therefore, $p \leq q$.

On the other hand, assume $p \leq q$. It suffices to prove that $I^{-} (p) \subset I^{-} (q)$. If $x \ll p$, by the push-up property we get $x \ll q$. Therefore, $I^{-} (p) \subset I^{-}(q)$. 
\end{proof}

So far, we have exhibited a metric topology, as well as chronological and causal relations in the future causal completion $\hat{M}$. In the remaining of this section we will be dealing with the construction of an adequate time separation $\hat{\tau}$  for $\hat{M}$. As a first step, notice that if $Q \in \hat{M}$ is generated by a future-directed chain $\{q_n\}$  and $p \in Q$, then there exists $N \in \mathbb{N}$ such that for all $n\geq N$, $p \ll q_n \ll q_{n+1}$. By the reverse triangle inequality of $\tau$ we have
$$\tau (p, q_{n+1}) \geq \tau (p,q_n) + \tau (q_n, q_{n+1}) > \tau (p,q_n).$$
Hence, the sequence of real numbers $\{\tau (p,q_n) : p \ll q_n\}$ is strictly increasing. Consequently, in view of Proposition \ref{lemadesigualdad} we define $\htau$ as follows: for $p \in M$, $P,Q \in \hat{M}$
\begin{itemize}
\item[(i)] $\htau (P,Q) := 0$ if $P \in \hat{\partial} M$,
\item[(ii)] $\htau (I^{-} (p), Q) := \sup \{\tau (p,q_n): Q = I^{-} (\{ q_n \})\}.$
\end{itemize}

Since the chain $\{q_n\}$ that generates a past set $Q$ might not be unique, we first have to show that $\htau$ is well defined. First note that if $\htau (P,Q) = 0$ there is nothing to prove.

%\begin{figure}[ht!]
%\centering \includegraphics[scale=0.4]{independencia.png}
%\end{figure}

Alternatively, If $\{q_n\}$ and $\{r_m\}$ are two different future-directed chains generating $Q$. It is enough to show that
$$\sup \{\tau(p,q_n) : Q = I^{-} (\{q_n\})\} = \sup \{\tau(p,r_m) : Q = I^{-} (\{r_m\})\}.$$
For any $p \in Q$ there exists $q_n \in Q$ with $p \ll q_n$, and hence there exists $r_m \in Q$ with $q_n \ll r_m$. Thus $p \ll q_n \ll r_m$ and by the reverse triangle inequality
$$\sup \{ \tau (p,r_m) \} \geq \tau (p,r_m) \geq \tau (p,q_n) + \tau (q_n, r_m) > \tau (p,q_n).$$
Hence $\sup \{ \tau (p,r_m) \}$ is an upper bound for $\{\tau (p, q_n)\}$ and thus 
$$\sup \{\tau(p,q_n) : Q = I^{-} (\{q_n\})\} \leq \sup \{\tau(p,r_m) : Q = I^{-} (\{r_m\})\}.$$
Similarly,
$$\sup \{\tau(p,r_m) : Q = I^{-} (\{r_m\})\} \leq \sup \{\tau(p,q_n) : Q = I^{-} (\{q_n\})\}.$$

We now notice that $\htau$ extends $\tau$ to the future causal completion $\hat{M}$.

\begin{proposition}\label{htauextiende}
If $P=I^-(p)$ and $Q=I^-(q)$, then $\htau (P,Q) = \tau (p,q)$.
\end{proposition}

\begin{proof}
By definition
$$\htau (P,Q) = \sup \{\tau (p,q_n) : Q = I^{-} (\{q_n\})\}.$$
Since $Q = I^{-} (q)$ is a PIP, the chain $\{q_n\}$ must converge to $q$. Moreover, the sequence $\{\tau (p,q_n)\}$ is non decreasing and since $(M,g)$ is globally hyperbolic, $\tau$ is a continuous function. Therefore
$$\htau (P,Q) = \sup \{ \tau (p,q_n): Q = I^{-}(\{q_n\}) \} = \lim_{n \to \infty} \tau (p,q_n) = \tau (p,q).$$
\end{proof}

The following series of lemmas are intended to show that $\htau$ defines a time separation function on $(\hat{M},\hat{\mathcal{T}}_c)$ when $(M,g)$ is a globally hyperbolic spacetime.

\begin{lemma}[Positivity]\label{clasificacausalidad}
If $P,Q \in \hat{M}$, then $\htau (P,Q) > 0$ if and only if $P \hat{\ll} Q$.
\end{lemma}

\begin{proof}
We begin assuming that $P \hat{\ll} Q$. Then $P = I^{-}(p)$ for some $p\in M$ by Proposition \ref{lemadesigualdad}.  By definition, there exists $q \in Q\setminus I^{-}(p)$ such that $I^{-}(p) \subset I^{-} (q)$, then $p \leq q$. Moreover, if $\{q_n\}$ is a future directed chain generating $Q$ then there exists $q_m \in \{q_n\}$ such that $p \leq q \ll q_n$ for all $n\ge m$ which implies  by the reverse triangle inequality
$$\tau (p,q_n) \geq \tau (p,q) + \tau(q,q_n) > 0.$$
Therefore
$$\hat{\tau}(I^{-}(p), Q ) = \sup \{\tau(p,q_n) :  Q = I^{-} (\{q_n\})\} > 0.$$

Conversely, assume that $\htau (P,Q) > 0$. Which by the definition of $\htau$ means that $P$ is a PIP. Suppose that $P = I^{-}(p)$. If $\htau (I^{-}(p) , Q) < +\infty$, for any $\epsilon > 0$ there exists $q_n$ with
$$\tau (p,q_n) > \hat{\tau}(I^{-}(p),Q) - \epsilon.$$
If we take $\epsilon = \frac{\hat{\tau}(I^{-}(p),Q)}{2}$, we have that $\tau(p,q_n) > 0$ which implies $p \ll q_n \ll q_{n+1}$. Hence, $q_n \in Q \setminus I^{-}(p)$ with $I^{-}(p) \subset I^{-}(q_n)$ and then $I^{-}(p) \hat{\ll} Q$. Now, if $\htau(I^{-} (p), Q) = +\infty$, for any $N \in \mathbb{N}$ there exists $q_m \in \{q_n\}$ such that
$$\tau (p,q_m) \geq N > 0,$$
which implies $p \ll q_n$ for all $n\geq m$. Then $I^{-}(p) \hat{\ll} Q$.
\end{proof}

The following result is needed in the proof of the reverse triangle inequality.

\begin{lemma}\label{lema:rev}
If $p \notin Q$ or $p \in \partial Q$, then $\htau (I^{-} (p), Q) = 0$.
\end{lemma}

\begin{proof}
Let $\{q_n\}$ be a future-directed chain generating $Q \in \hat{M}$. Proceeding by contradiction, let us suppose that $\htau (I^{-} (p), Q ) > 0$ then there exists $q_k \in \{ q_n \}$ with $\tau (p,q_k) > 0$, but this only occurs if $p \ll q_k$, which is a contradiction because $p \notin Q$ or $p \in \partial Q$, and in either case, there can be no such $q_k$.
\end{proof}

\begin{lemma}[Reverse triangle inequality]\label{desigualdadtrianguloinvertida}
If $P,Q,R \in \hat{M}$ are such that $P \hat{\leq} Q \hat{\leq} R$, then 
$$\htau (P,R) \geq \htau (P,Q) + \htau (Q,R).$$
\end{lemma}

\begin{proof}
We proceed  case by case. If $P,Q,R$ are all TIPs then
$$\htau (P,R) = 0 \geq \htau(P,Q) + \htau(Q,R) = 0.$$
Moreover, if $P$ is a TIP by Proposition \ref{lemadesigualdad}, $Q$ and $R$ are TIPs as well.
Also, if $P,Q,R$ are all PIPs, that is, $P = I^{-}(p), Q= I^{-}(q)$ and $R = I^{-}(r)$ for some $p,q,r \in M$ then the reverse triangle inequality of $\tau$ gives us the result.
Thus, we have only two cases left, namely
\begin{itemize}
\item[(i)] $P \hat{\leq} Q \hat{\leq} R$ with $P=I^-(p)$, $Q=I^-(q)$, $R \in \hat{\partial} M$;
\item[(ii)] $P \hat{\leq} Q \hat{\leq} R$ with $P= I^{-} (p)$, $Q$, $R \in \hat{\partial} M$.
\end{itemize}

For case (i), notice that if $p\in\partial R$ then $q\not\in R$. Moreover, if
$p$, $q\in \partial R$, then $p\rightarrow q$\footnote{Recall that the \emph{horismos} relation $p\rightarrow q$ is defined as $p\leq q$ and $p\not\ll q$.} follows. Thus we only have the following possibilities: 
\begin{itemize}
\item[(a)] $p \rightarrow q$ and $p, q \in \partial R$, 
\item[(b)] $p \ll q$ with $q \in \partial R$,
\item[(c)] $p \ll q$ and $q \in R$,
\item[(d)] $p \rightarrow q$, $p \in R$ and $q \in \partial R$,
\item[(e)] $p \rightarrow q$ and $p,q\in R$. 
\end{itemize}

In case (a) observe that $\tau(p,q)=0$, $\hat{\tau}(I^{-}(p),R)=\hat{\tau}(I^{-}(q),R)=0$ in virtue of Lemma \ref{lema:rev}. Thus the triangle inequality holds trivially. 

Now, for case $(b)$  consider $R=I^{-}(\{r_{m}\})$. Since
 $p \in R$ and $\hat{\tau}(I^{-}(q),R)=0$,  take a future directed timelike sequence $\{q_{n}\}$ such that $q_{n} \in
I^{-}(q)$ and $q_{n} \rightarrow q$. Thus, $ p \ll q_{n} \ll r_{m_n}$ for some $n$ and $m_n$ natural numbers and this leads to
\[
\tau(p,r_{m_n}) \geq \tau(p,q_n) + \tau(q_n,r_{m_{n}}) \geq \tau(p,q_n)
\]
by the reverse triangle inequality of $\tau$. Since $(M,g)$ is globally hyperbolic, the time separation function $\tau$  is continuous, thus 
$$\hat{\tau}(I^-({p)},R) \geq \hat{\tau}(I^{-}(p),I^{-}(q))+0= \hat{\tau}(I^{-}(p),I^{-}(q))+\hat{\tau}(I^{-}(q),R).$$

For (c) it is enough to prove
$$\htau (I^{-} (p), R) \geq \tau (p,q) + \htau (I^{-}(q) , R).$$
Let $\{r_m\}$ be a future-directed chain generating $R$. By definition
\begin{align*}
\htau (I^{-} (p),R) :=& \sup \{\tau (p,r_m): R= I^{-} (\{r_m\})\} \\
\htau (I^{-} (q),R) :=& \sup \{\tau (q,r_m): R= I^{-} (\{r_m\})\}
\end{align*}
We know that $I^{-}(p) \hat{\leq} I^{-} (q)$ if and only if $p \leq q$. Therefore $p\leq q \ll r_m$ for large $m$ and  we have
$$\tau (p,r_m) \geq \tau(p,q) + \tau (q,r_m),$$
which in turn implies
$$\htau ( I^{-} (p), R) \geq \tau (p,q) + \htau (I^{-}(q), R).$$

In case $(d)$ we have that  
$\hat{\tau}(I^{-}(q),R)=0=\hat{\tau}(I^{-}(p),I^{-}(q))$ and since $p \in R$ we have that 
$p \ll r_{m}$ for some $m \in \mb{N}$ which leads to $\tau(p,r_{m})>0$ and this gives $\hat{\tau}(I^{-}(p),R)>0=\hat{\tau}(I^{-}(p),I^-(q))+\hat{\tau}(I^{-}(q),R)$.

In order to prove $(e)$ note that $\hat{\tau}(I^{-}(p),I^{-}(q))=0$ and following the same argument in point $(c)$ we have that $p \leq q \ll r_{m}$ for  large $m$ and thus $\tau(p,r_{m}) \geq \tau(q,r_{m})$ by the inverse triangle inequality in $(M,g)$. Thus, we have $\hat{\tau}(I^{-}(p),R) \geq \hat{\tau}(I^{-}(q),R)$ which is the reverse triangle inequality in this case.

Similarly, for case (ii) notice that $R\not\in\hat{I}^+(Q)$, thus we have the following possibilities:

\begin{itemize}
	\item[(a)] $p \in \partial Q \cap \partial R$, 
	\item[(b)] $p \in \partial Q$ and $p \in R$,
	\item[(c)] $p \in Q \cap R$. 
\end{itemize} 

In  case (a) observe that all the quantities involved are zero and thus the reverse triangle inequality holds trivially. 

In case $(b)$ we have that 
$\hat{\tau}(I^{-}(p),Q)=0$ and $\hat{\tau}(Q,R)=0$ and since $p \in R$ we have that for any future directed sequence $\{r_m\}$ that generates $R$ we have 
$p \ll r_{m}$ for $m \geq m_{0}$. The latter implies that $\hat{\tau}(I^{-}(p),R) \geq \tau(p,r_{m})>0$ and thus reverse triangle inequality holds trivially again. 

Finally, for (c) we have to prove that
$$\htau(I^{-} (p) , R) \geq \htau (I^{-}(p) , Q)$$
since $\htau (Q,R) = 0$, by definition. In order to verify this, take $\{q_n\}$ and $\{r_m\}$ future-directed chains generating $Q$ and $R$, respectively.  Observe that $Q \subset R$ since they are causally related in $\hat{M}$, so there are $n, m_{n} \in \mb{N}$ such that 
$p \ll q_{n} \ll r_{m_{n}}$. Therefore, we have $\tau (p,r_{m_{n}}) \geq \tau (p,q_n)$ for large $n$,
and thus
$\htau (I^{-}(p) , R) \geq \htau (I^{-}(p) , Q)$ holds.
\end{proof}

All that is left for $\htau$ to be a time separation function is for it to be lower semicontinuous. 

\begin{lemma}[Lower semicontinuity]\label{lowersemicontinua}
The function $\htau$ is lower semicontinuous in $\hat{M}$, i.e. for all $P,Q \in \hat{M}$, and any $\delta > 0$, if $P_n$ and $Q_n$ are sequences that converge to $P$ and $Q$, respectively, then there exists $N \in \mathbb{N}$ such that
$$\htau (P_n , Q_n) > \htau (P,Q) - \delta$$
for all $n\geq N$.
\end{lemma}

\begin{proof}
If $\htau(P,Q) = 0$, there is nothing to prove. Otherwise, if $\htau (P,Q) > 0$, then $P \hat{\ll} Q$  and $P = I^{-} (p)$ for some $p\in M$ by Proposition \ref{lemadesigualdad} . Let $\delta > 0$ and $\{q_k\}$ be a future-directed chain generating $Q$, and for every $n$, let $\{q^n _k\}$ be a future-directed chain generating $Q_n$. Given that $P_n$ converges to $P$ and $\hat{\partial} M$ is closed in virtue of Theorem \ref{teo:clt}, there exists $N_1 \in \mathbb{N}$, such that for all $n \geq N_1$, $P_n$ is a PIP, that is, $P_n = I^{-}(p_n)$ with $p_n \in M$ and $p_n$ converges to $p$.

Since $p_n$ converges to $p$, there exists $N_2 \in \mathbb{N}$ such that for all $n \geq N_2$, $I^{-} (p_n) \in \hat{I}^{-}(Q)$, since the latter is an open neighborhood for $I^{-}(p)$. Thus, for all $n \geq N_2$ there exists $k$ such that $p_n \ll q_k \in Q$. Moreover, since $Q = \liminf (Q_n) = \limsup (Q_n)$, then every $q_k$ is in all but a finitely many of the sets $Q_n$. Therefore, there exists $q^{n}_{r}$ with $p_n \ll q_k \ll q^{n}_{r}$ and as a consequence, by the reverse triangle inequality $\tau (p_n, q^{n}_{r}) > \tau (p_n, q_k).$
On the other hand, using that $\tau$ is lower semicontinuous, we have that there exists $N_3 \in \mathbb{N}$ such that for all $n \geq N_3$ 
$$\tau (p_n, q_k) > \tau (p, q_k) - \frac{\delta}{2}.$$
Then, if $n \geq  N = \max \{N_1,N_2,N_3\}$
$$\tau (p_n, q^{n}_{r}) > \tau (p, q_k) - \frac{\delta}{2},$$
whence $\tau (p_n, q^{n}_{r}) + \frac{\delta}{2}$ is an upper bound for $\tau (p, q_k)$. Which implies
$$\tau (p_n, q^{n}_{r}) \geq \sup \{\tau (p, q_k): Q = I^{-}(\{q_k\}) \} - \frac{\delta}{2}.$$
By definition of least upper bound
$$\sup \{\tau (p_n, q^{n}_{r}) : Q_n = I^{-}(\{q^{n}_r\}) \} > \sup \{\tau (p, q_k): Q = I^{-}(\{q_k\}) \} - \delta,$$
which by definition of $\htau$ means that for all $n \geq N$ 
$$\htau (P_n , Q_n) > \htau (P,Q) - \delta.$$
\end{proof}

The following theorem summarizes the results of the previous lemmas.

\begin{theorem}\label{teoremaprincipal}
Let $(M,g)$ be a globally hyperbolic space-time. Then $(\hat{M}, d_c, \hat{\ll}, \hat{\leq} , \htau)$ is a Lorentzian pre-length space.
\end{theorem}

It is important to note that the relations $\hat{\ll}$, $\hat{\leq}$ as well as the CLT topology are not affected by a conformal change on the spacetime metric $g$. However, this is not the case for the time separation $\tau$. Thus, it is expected to have different Lorentzian pre-length structures in the future causal completion within the same conformal class.

%%%%%%%%%%%%%
%%%%%%%%%%%%%
%%%%%%%%%%%%%

\section{Applications}\label{sec:warped}

In this section we show that a class of warped product spacetimes satisfies the condition of Thrm \ref{teo:opagree} and hence its associated chronological and CLT topologies coincide.  This result can be used to carry out explicit calculations regarding the pre-length structure of the future causal completion.\footnote{A thorough description of the future causal boundary of such spacetimes can be found in \cite{valana}.} Finally, as an illustrative example we consider the particular case of de Sitter spacetime.

Recall that a Lorentzian warped product is a manifold $V=\mb{R} \times M$ furnished with a metric of the form 
$$g=-dt^{2} + \alpha(t) h,$$ 
where, $(M,h)$ is a Riemannian manifold and $\alpha$ is a smooth positive function over $\mb{R}$.

The chronological relation on these spacetimes can be characterized as follows (see \cite[Section 2.2]{joniluis}):
		\begin{lemma}
			Let $(V,g)$ be a Lorentzian warped spacetime. If $(t_{0},x_{0}), (t_{1},x_{1})$ points in $V$, then, 
			\begin{equation*}
				(t_{0},x_{0}) \ll (t_{1},x_{1}) \Leftrightarrow d(x_{0},x_{1}) < \int_{t_{0}}^{t_{1}} \frac{ds}{\sqrt{\alpha(s)}}
			\end{equation*}
		\end{lemma}
		
 The future causal completion can be characterized depending on the value of $\int_{0}^{+\infty} \frac{ds}{\sqrt{\alpha(s)}}$. Indeed we have\footnote{ Here $\partial_\mathcal{B} M$ and $\partial_CM$ are the Busemann and (metric) Cauchy boundaries or $M$. Refer to \cite{joniluis,valana,isobound} for a detailed account on the structure and topology of $\hat{V}$.}
		
\begin{theorem}\label{XX}
Let $V=\mb{R} \times M$ with $g=-dt^2+\alpha(t) h$ be a warped spacetime and $(M,d)$ a locally compact metric space. Then, 
\begin{enumerate} \setlength\itemsep{1em}				
\item If $\int_{0}^{+\infty} \frac{ds}{\sqrt{\alpha(s)}}= \infty$ ,then, the future causal boundary $\hat{\partial} V$ is an infinite null cone with base $\partial_{\mathcal{B}}M \setminus \partial_{C} M$ with apex in $i^{+}$ and timelike lines over each point in $\partial_{C} M$ and final point in $i^{+}$. Moreover, $\hat{V}$ is homeomorphic to $V \cup (\mb{R} \times \partial_{C} M) \cup (\mb{R} \times \partial_{\mathcal{B}} M) \cup i^{+}$.

\item If $\int_{0}^{+\infty} \frac{ds}{\sqrt{\alpha(s)}} < \infty$, then, $\hat{\partial} V$ is a copy of the Cauchy completion $M^{C}$ of $(M,h)$ and timelike lines over each point in $\partial_{C} M$ that finishes in the same point at the copy at infinity of $M^{C}$. Moreover, $\hat{V}$ is homeomorphic to $V \cup (\mb{R} \times \partial_{C} M) \cup (\{\infty\} \times M^{C})$. 
\end{enumerate}
\end{theorem}

Let $(V,g)$ be a globally hyperbolic warped product spacetime. Then the Riemannian manifold $(M,h)$ is complete as can be seen in \cite[Theorem 3.68]{beem} and hence the Cauchy boundary $\partial_{C} M = \emptyset$. Therefore, if in addition 
$$\int_{0}^{+\infty} \frac{ds}{\sqrt{\alpha(s)}} < \infty ,$$
then by the previous theorem we have that the future causal completion is characterized as $\hat{V} \equiv V \cup (\{+\infty\} \times M)$ and  the causal boundary consists of a copy of $M$ at infinity. As a consequence, any TIP $P \in \hat{\partial} V$ can be identified with a set of the form $P = I^{-}(+\infty, x)$, $x \in M$.\footnote{Here, $I^{-}(+\infty,x):=\{(t_{0},x_{0})  \in V \mid \int_{t_{0}}^{+\infty} \frac{ds}{\sqrt{\alpha(s)}}< \int_{0}^{+\infty} \frac{ds}{\sqrt{\alpha(s)}} - d(x_{0},x) \}$}

\begin{proposition}\label{crono-es-hausdorff}
Let $(V,g)$ be a globally hyperbolic warped product spacetime with $\int_{0}^{+\infty} \frac{ds}{\sqrt{\alpha(s)}} < \infty$. Then the chronological topology $\hat{\mathcal{T}}_{chr}$ in $\hat{V}$ is Hausdorff.
\end{proposition}

\begin{proof}
Recall that in $i(V)$, the chronological topology agrees with the topology of the space-time $V$, which is Hausdorff. Thus there are only two cases left to consider: 
\begin{enumerate} 
\item $P$ is PIP and $Q$ is a TIP,
\item $P$ and $Q$ are both TIPs.
\end{enumerate}

\begin{figure}[ht!]
\centering \includegraphics[scale=0.2]{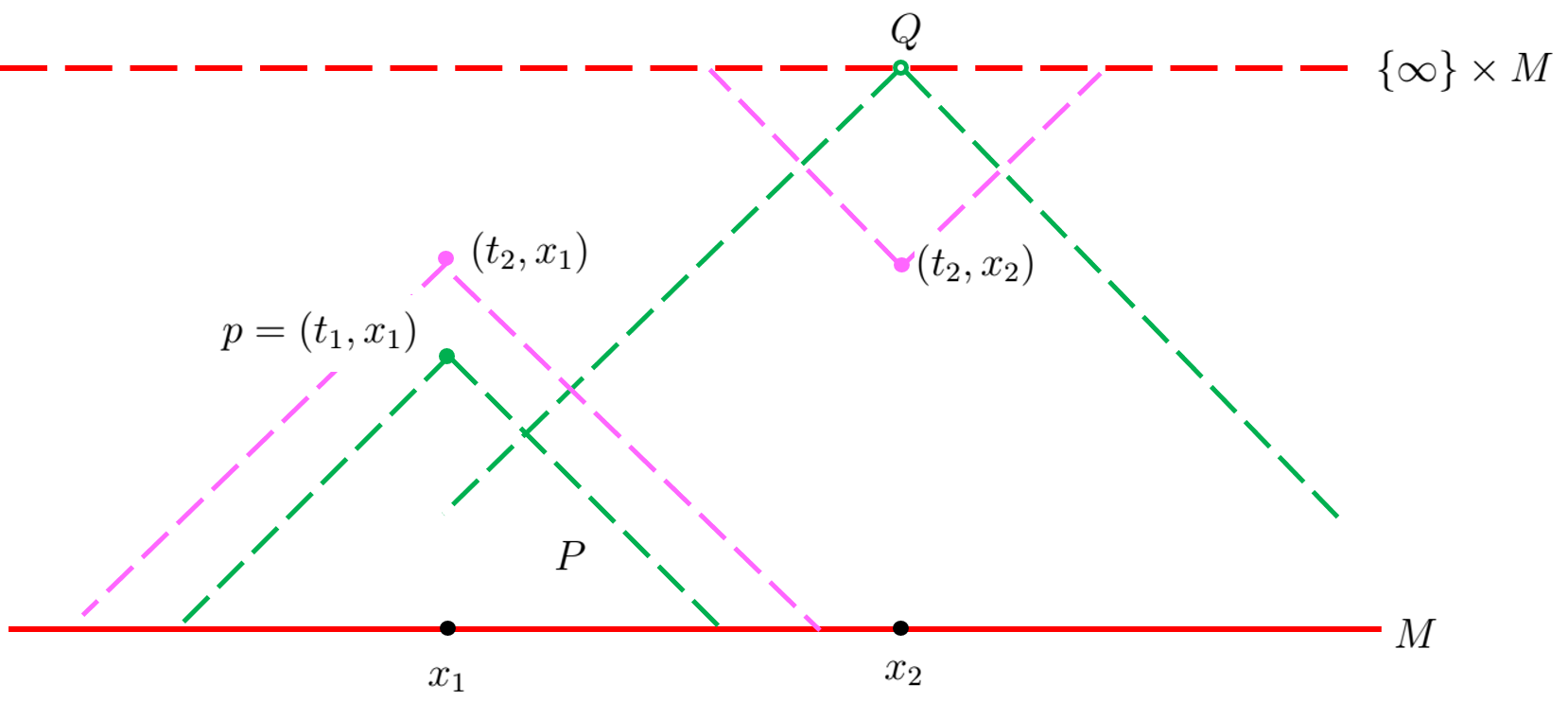}
\centering{\caption{Case (1). $P$ is proper and $Q$ is terminal.}}
\end{figure}

\textbf{Case 1.} Consider $p = (t_1 , x_1)$ a point in $V$ and $P = I^{-} (p)$ its correspondent past set; and take $Q = (+\infty , x_2) \in \hat{\partial} V$. We know, by the definition of the chronological relation in $V$, that for $t_2 > t_1$, we have
$$(t_1 , x_1) \ll (t_2 , x_1) $$
because
$$ 0 = d (x_1, x_1) < \int_{t_1}^{t_2} \frac{ds}{\sqrt{\alpha (s)}}.$$
Note that $(t_2,x_2) \in Q$ since $(t_2 ,x_2) \ll (t, x_2)$ for all $t>t_2$. Hence, there exists $q_n \in Q$ with $(t_2 , x_2) \ll q_n$ and as a consequence ${I}^{-}(t_2,x_2) \hat{\ll} Q$. Therefore $Q \in \hat{I}^{+} (I^{-} (t_2, x_2)).$ On the other hand, since $(t_1, x_1) \ll (t_2,x_1)$ then $P \hat{\ll} I^{-} (t_2,x_1)$, and by definition $P \in \hat{I}^{-} ( I^{-} (t_2,x_1))$. Now we need to prove
$$\hat{I}^{+} (I^{-} (t_2, x_2)) \cap \hat{I}^{-} ( I^{-} (t_2,x_1)) = \emptyset.$$
By contradiction, suppose there exists $R = I^{-} (r,y) \in \hat{V}$, with $r \in \mathbb{R}$, such that 
$$R \in \hat{I}^{+} (I^{-} (t_2, x_2)) \cap \hat{I}^{-} ( I^{-} (t_2,x_1)) = \emptyset.$$
Note that $r < \infty$ because otherwise $R$ and $Q$ would be two TIPs with $R \hat{\ll} Q$, which cannot occur since $V$ is globally hyperbolic. Thus, $R$ is proper and
$$I^{-} (t_2, x_2) \hat{\ll} I^{-} (r,y) \hat{\ll} I^{-} (t_2, x_1)$$
which happens by Proposition \ref{relacionesextendidas} if and only if
$$(t_2,x_2) \ll (r,y) \ll (t_2, x_1),$$
and by transitivity
$$(t_2, x_2) \ll (t_2, x_1),$$
that is,
$$0 \leq d(x_2,x_1) < \int_{t_2}^{t_2} \frac{ds}{\sqrt{\alpha (s)}} = 0,$$
a contradiction. This concludes case 1.

\begin{figure}[ht!]
\centering \includegraphics[scale=0.2]{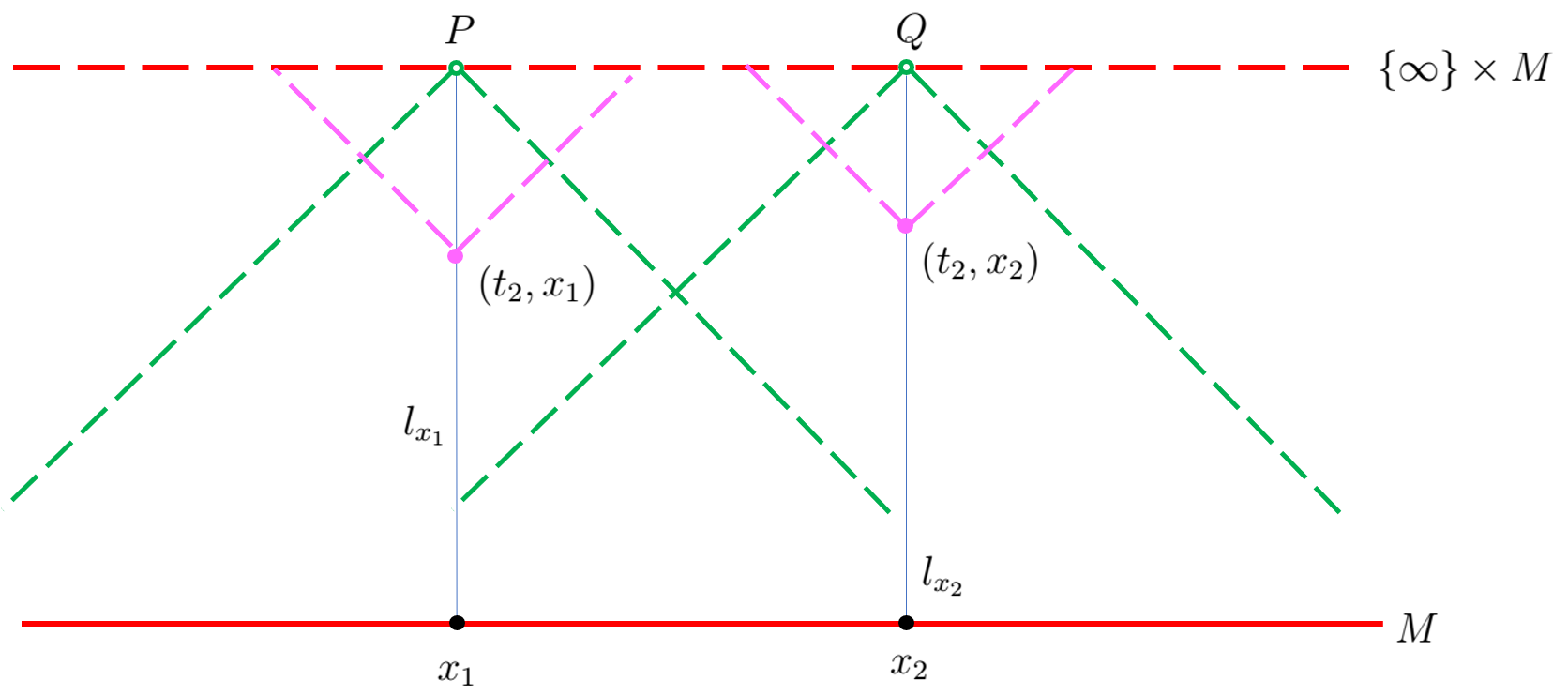}
\caption{Case 2. $P$ and $Q$ are both terminal.}
\end{figure}

\textbf{Case 2.} Consider $P,Q \in \hat{\partial}V$ with $P \neq Q$. We can represent these sets as
$$P = I^{-} (\infty, x_1), \qquad Q = I^{-} (\infty , x_2).$$
Since $\int_{0}^{\infty} \frac{ds}{\sqrt{\alpha (s)}} < +\infty$ we can take $t_1,t_2 \in \mathbb{R}$ such that
$$\int_{t_1}^{\infty} \frac{ds}{\sqrt{\alpha (s)} }< \frac{1}{3} d(x_1, x_2) \quad \text{and} \quad \int_{t_2}^{\infty} \frac{ds}{\sqrt{\alpha (s)} }< \frac{1}{3} d(x_1, x_2)$$
which implies $(t_1, x_1) \notin Q$ and $(t_2,x_2) \notin P$, respectively.

Without loss of generality suppose $t_1 < t_2$. We know that
$$P \in \hat{I}^{+} (I^{-} (t_2, x_1)), \qquad Q \in \hat{I}^{+} (I^{-} (t_2, x_2)),$$
so it is enough to prove that this sets are disjoint, i.e.
$$\hat{I}^{+} (I^{-} (t_2, x_1)) \cap \hat{I}^{+} (I^{-} (t_2, x_2)) = \emptyset.$$
We proceed by contradition. Take $R = I^{-} (r,y) \in \hat{V}$ such that 
$$R \in \hat{I}^{+} (I^{-} (t_2, x_1)) \cap \hat{I}^{+} (I^{-} (t_2, x_2)),$$
whence
$$ I^{-} (t_2, x_1) \hat{\ll} I^{-} (r,y) \quad \text{and} \quad I^{-} (t_2, x_2) \hat{\ll} I^{-} (r,y) $$
by Proposition \ref{relacionesextendidas} if and only if
$$ (t_2,x_1) \ll (r,y) \quad \text{and} \quad (t_2,x_2) \ll (r,y)$$
which happens if and only if
$$d(x_1, y) < \int_{t_2}^{r} \frac{ds}{\sqrt{\alpha (s)}} \quad \text{and} \quad d(x_2, y) < \int_{t_2}^{r} \frac{ds}{\sqrt{\alpha (s)}}$$
where adding both inequalities and using the triangle inequality for $d$ we get
$$d (x_1, x_2) \leq d(x_1,y) + d(x_2,y) < 2 \int_{t_2}^{r} \frac{ds}{\sqrt{\alpha (s)}}$$
and so
$$\int_{t_2}^{\infty} \frac{ds}{\sqrt{\alpha (s)}}  < \int_{t_2}^{r} \frac{ds}{\sqrt{\alpha (s)}},$$
where, being $r$ finite or infinite, we get a contradiction.
\end{proof}

The next result is an immediate consequence of Proposition \ref{crono-es-hausdorff} and Theorem \ref{teo:opagree}.

\begin{coro}
Let $(V,g)$ be a globally hyperbolic warped product spacetime with $\int_{0}^{+\infty} \frac{ds}{\sqrt{\alpha(s)}} < \infty$. Then $\hat{\mathcal{T}}_{chr}=\hat{\mathcal{T}}_{c}$.
\end{coro}

%\begin{remark}\label{longinfinita}
%\footnote{Note that  if $\gamma: [t_1,t_2] \to V$ is a vertical curve, that is, $\gamma (t) = ( t , a)$ with $a \in M$ fixed, then $L_g (\gamma) = \int_{t_1}^{t_2} \mid \dot{\gamma} \mid dt =  t_2 - t_1$. This indicates that the $g$-length of $\gamma$ grows to infinity as $t_2\to\infty$.}
%\end{remark}

An important fact to highlight  is that on these spacetimes, when $P = I^{-} (p)$ is a PIP and $Q$ is a TIP, their time separation is either zero or infinity. 

\begin{proposition}
If $P,Q \in \hat{V}$ with $Q \in \hat{\partial} V$ then either $\htau (P,Q) = 0$ or $\htau (P,Q) = \infty$.
\end{proposition}

\begin{proof}
If $P \hat{\not\ll} Q$ or $P \in \hat{\partial} V$ then $\htau (P,Q) = 0$ and we are done. Now let us assume that $ P = I^{-}(p) \hat{\ll} Q = I^{-} (\infty, a )$ and that $\{q_n\}$ is a future-directed chain generating $Q$. By chronology, there exists $q_N \in Q$ with $p \ll q_N \ll q_n$ for all $n > N$ and since $V$ is globally hyperbolic, $\tau (p,q_N) < \infty$. Given that $\htau$ does not depend on the election of the chain that generates Q, therefore we can choose a chain such that
$$q_n = (t_{n}, a),$$
which is based on the same spatial constant coordinate $a \in M$ as $Q$. From here, since $Q$ is terminal, $q_n$ does not converge in $V$ and then $t_{n} \to \infty$. This implies\footnote{Note that  if $\gamma: [t_1,t_2] \to V$ is a vertical curve, that is, $\gamma (t) = ( t , a)$ with $a \in M$ fixed, then $L_g (\gamma) = \int_{t_1}^{t_2} \mid \dot{\gamma} \mid dt =  t_2 - t_1$. This indicates that the $g$-length of $\gamma$ grows to infinity as $t_2\to\infty$.} 
$$\tau (p,q_n) \geq \tau (p, q_N) + \tau (q_N, q_n) > \tau (q_N,q_n) \geq t_n - t_N.$$
Thus $\htau(P,Q) = \infty$ for any $P$ a PIP and $Q$ a TIP chronologically related.
\end{proof}

We end up this section providing an explicit example. Consider de Sitter spacetime $\mathbb{S}^4_1$, the Lorentzian spaceform of constant sectional curvature $K=1$. It can be realized as the hyperboloid
$$-v^2 + w^2 + x^2 + y^2 + z^2 = 1$$
in flat five-dimensional Minkowski space $\mathbb{R}^5_1=(\mathbb{R}^5,L)$. Recall $\mathbb{S}^4_1$ can also be described as the globally hyperbolic warped product
$$(\mathbb{R} \times \mathbb{S}^3, -dt^2 + \cosh ^2 (t) g_{\mathbb{S}^3}),$$
where $g_{\mathbb{S}^3}$ is the standard round metric on $\mathbb{S}^3$. Thus, $$\int_{0}^\infty \frac{ds}{\sqrt{\alpha(s)}} = \int_{0}^\infty \frac{ds}{\cosh (s)} = \frac{\pi}{2} < +\infty.$$
Moreover, by Theorem \ref{XX} the future causal boundary $\hat{\partial} \mathbb{S}^4_1$ is a copy at infinity of the base manifold $\mathbb{S}^3$, a picture that agrees with its standard future conformal infinity $\mathcal{I}^+$ \cite{HW}.

Let $p\leq q$ be two causally related points in $\mathbb{S}^4_1$. Since the time separation function is given by
$$\cosh (\tau (p,q)) = L(p,q),$$
then 
$$\tau (p,q) = t_2 - t_1$$
for $p=(t_1,x)$, $q=(t_2,x)$ two points on the same geodesic normal to $\mathbb{S}^3$. Thus, if $t_2 \to \infty$ then $\tau (p,q) \to \infty$. If we consider a PIP $P = I^{-} (t,x)\in \hat{\mathbb{S}}^4_1$  and a TIP $Q = I^{-} (\infty , y) \in \hat{\partial} \mathbb{S}^4_1$ a TIP, then by definition
$$\htau (P,Q) = \sup \{\tau (p,q_n): Q = I^{-}(\{ q_n \})\},$$
where $\{ q_n\}$ is a future-directed chain generating $Q$. Choose a chain such that
$$q_n = (t_{n}, y),$$
which is based on the same spatial coordinate $y \in \mathbb{S}^3$. From here, since $Q$ is terminal, $q_n$ does not converge in $\mathbb{S}^4_1$ and then $t_{n} \to \infty$.% This implies
%$$\htau (P,Q) = \infty,.$$
%for any $P$ a PIP and $Q$ a TIP based on the same point $x \in \mathbb{S}^3$.

%Now consider $R \in \hat{\partial}\mathbb{S}^4_1$ different from $Q$. 
We know that if $p=(t,x) \notin Q$ or $p \in \partial Q$ we have $\htau (I^{-} (p), Q) = 0$ by Lemma \ref{lema:rev}. On the other hand, if $p \in Q$, there exists $N \in \mathbb{N}$ with $p \ll q_N \ll q_n$ for all $n > N$. By the reverse triangle inequality for $\tau$ we have
$$\tau (p,q_n) \geq \tau (p,q_N) + \tau(q_N,q_n) > \tau(q_N,q_n) \geq t_n - t_N,$$
where the second term on the right-hand-side goes to infinity as $n \to \infty$. Therefore
$$\htau (P,Q) = \infty.$$

In summary, the time separation function $\htau$ for the future causal completion of de Sitter spacetime is given by
\begin{equation*}
\htau (P,Q) =\left\lbrace\begin{array}{ccc}
\tau (p,q) & \text{if} & P = I^{-}(p), Q= I^{-}(q) \\
0 & \text{if} & P \hat{\nleq} Q \\
\infty & \text{if} & P \hat{\ll} Q \text{ and } Q \in \hat{\partial} \mathbb{S}^4_1.
\end{array}\right.
\end{equation*}

\section*{Acknowledgements} The authors would like to thank the organizing committee of SCRI21, for providing a source of great mathematical ideas for years to come. L. Ake Hau was partially supported by Conacyt under grants SNI 367994 and Estancia Posdoctoral por M\'exico 854544. S. Burgos acknowledges the support of Conacyt under the Becas Nacionales program (1089466). D. Solis was partially supported by Conacyt under grant SNI 38368.

%\section*{Declarations}

%\begin{itemize}
%\item \textbf{Funding}  L. Ake Hau was partially supported by Conacyt under grants SNI 367994 and Estancia Posdoctoral por M\'exico 854544. S. Burgos acknowledges the support of Conacyt under the Becas Nacionales program (1089466). D. Solis was partially supported by Conacyt under grant SNI 38368.
%\item \textbf{Conflict of interest/Competing interests} The author declare they have no conflict of interests.  
%\item \textbf{Availability of data and materials} No data was collected nor used in this work.
%\end{itemize}

\vspace{1cm}

\noindent \textbf{Luis Ak\'e Hau}. Facultad de Matem\'aticas, Universidad Aut\'onoma de Yucat\'an, Perif\'erico Norte Tablaje 13615,  C.P. 97110, M\'erida, M\'exico. \\
luis.ake@correo.uady.mx

\vspace{.3cm}

%( \Letter \ )

\noindent \textbf{Saul Burgos}. Facultad de Matem\'aticas, Universidad Aut\'onoma de Yucat\'an, Perif\'erico Norte Tablaje 13615,  C.P. 97110, M\'erida, M\'exico. \\
saul.burgos@alumnos.uady.mx

\vspace{.3cm}

\noindent \textbf{Didier A. Solis}.  Facultad de Matem\'aticas, Universidad Aut\'onoma de Yucat\'an, Perif\'erico Norte Tablaje 13615,  C.P. 97110, M\'erida, M\'exico. \\
didier.solis@correo.uady.mx

%%Default %%
%\input sn-sample-bib.tex%

\end{document}